\documentclass[submission,copyright,creativecommons]{eptcs}
\providecommand{\event}{SR 2014} 
\usepackage{breakurl}             

\usepackage[fleqn]{amsmath}
\usepackage{amssymb,latexsym,wasysym,dsfont,stmaryrd}
\usepackage{relsize}
\usepackage{xspace}
\usepackage{graphicx}
\usepackage{gastex}
\usepackage{tikz}
\usetikzlibrary{arrows,automata,shapes.geometric,decorations.pathmorphing,backgrounds,positioning,fit,petri,calc}
\usepackage{wrapfig}
\usepackage{mymacros,sectionningmacros}

\title{Automata Techniques for Epistemic Protocol Synthesis}
\author{Guillaume Aucher
\institute{Universit\'e de Rennes 1 - INRIA\\ Rennes, France}
\email{guillaume.aucher@irisa.fr}
\and
Bastien Maubert
\institute{Universit\'e de Rennes 1\\
Rennes, France}
\email{bastien.maubert@irisa.fr}
\and
Sophie Pinchinat
\institute{Universit\'e de Rennes 1\\
Rennes, France}
\email{sophie.pinchinat@irisa.fr}
}
\def\titlerunning{Automata Techniques for Epistemic Protocol Synthesis}
\def\authorrunning{G. Aucher, B. Maubert \& S. Pinchinat}
\begin{document}
\maketitle

\begin{abstract}
  In this work we aim at applying automata techniques to problems
  studied in Dynamic Epistemic Logic, such as epistemic planning. To
  do so, we first remark that repeatedly executing \emph{ad infinitum}
  a propositional event model from an initial epistemic model yields a
  relational structure that can be finitely represented with
  automata. This correspondence, together with recent results on
  \emph{uniform strategies}, allows us to give an alternative
  decidability proof of the epistemic planning problem for
  propositional events, with as
  by-products accurate upper-bounds on its time complexity, and the
  possibility to synthesize a finite word automaton that describes the set of
  all solution plans. In
  fact, using automata techniques enables us to solve a much more general problem, that we
  introduce and call \emph{epistemic protocol synthesis}.
\end{abstract}

\section{Introduction}

Automated planning, as defined and studied in
\cite{DBLP:books/daglib/0014222}, consists in computing a finite
sequence of actions that takes some given system from its initial
state to one of its designated ``goal'' states. The Dynamic Epistemic
Logic (DEL) community has recently investigated a particular case of
automated planning, called \emph{epistemic planning}
\cite{DBLP:journals/jancl/BolanderA11,DBLP:conf/lori/LowePW11,DBLP:conf/ijcai/AucherB13}.
In DEL, epistemic models and event models can describe accurately
 how agents perceive the occurrence of events, and how their
knowledge or beliefs evolve. Given initial epistemic
states of the agents, a finite set of available events, and an
epistemic objective, the epistemic planning problem consists in
computing (if any) a finite sequence of available events whose 
occurrence results in a situation satisfying the objective property. While this problem
is undecidable in general
\cite{DBLP:journals/jancl/BolanderA11,DBLP:conf/ijcai/AucherB13},
restricting to \emph{propositional events} (those whose pre and
postconditions are propositional) yields decidability \cite{Yu2013}.

In this paper, preliminary to our main results we bring a new piece to
the merging of various frameworks for knowledge and time. Some
connections between DEL and Epistemic Temporal Logics (ETL) are
already known
\cite{DBLP:journals/synthese/HoshiY09,van2009merging,aucher2011exploring,DBLP:conf/ijcai/WangA13}.
We establish that structures generated by iterated execution of an event model from an epistemic model are regular structures, \ie
they can be finitely represented with
automata, in case the event model is propositional. 
This allows us to  reduce the epistemic planning problem for
propositional events to the \emph{uniform strategy problem}, as studied
in \cite{maubertIGTR2013,maubertSR2013,maubertFSTTCS2013court}. The automata techniques developed for uniform
strategies then provide an alternative proof of \cite{Yu2013}, with the
additional advantage of bringing accurate upper-bounds on the time
complexity of the problem, as well as an effective synthesis procedure
to generate the recognizer of all solution plans.  In fact, our
approach allows us to solve a generalized problem in DEL, that we call \emph{epistemic
  protocol synthesis problem}, and which is essentially the problem of
synthesizing a protocol from an epistemic temporal specification; its semantics  relies on the
interplay between DEL and ETL. We then make use of the connections
with regular structures and uniform strategies to solve this latter general
problem.

\section{DEL models}
\label{sec-DEL}

For this paper we fix $\agents$, a finite set of \emph{agents}, and
$\AP$ always denotes a finite set of atomic propositions (which 
is not fixed).  The epistemic language
$\langEL$ is simply the language of propositional logic extended with
``knowledge'' modalities, one for each agent. Intuitively, $\K_i\phi$
reads as ``agent $i$ knows $\phi$''. The syntax of $\langEL$ is given
by the following grammar: \[\phi::=p\mid \neg\phi \mid \phi\ou\phi
\mid \K_i\phi,\hspace{1cm} (\text{where } p\in\AP \text{ and } i\in\agents)\]

The semantics of $\langEL$ is given in terms of epistemic
models. Intuitively, a (pointed) epistemic model $(\epsmodel,\world)$
represents how the agents perceive the actual world $\world$.

\begin{definition}
\label{def-epsmodel}
An \emph{epistemic model}  is a tuple $\epsmodel=(\setworlds,\{\relworldsi\}_{i\in\agents},\valworlds)$ where
 $\setworlds$ is a  finite set of possible \emph{worlds},
 $\relworldsi\subseteq \setworlds\times \setworlds$
  is an
  \emph{accessibility relation} on $\setworlds$ for agent
  $i\in\agents$, and
 $\valworlds:\AP\rightarrow 2^\setworlds$ is a \emph{valuation function}.
\end{definition}

We write $\world\in \epsmodel$ for $\world\in \setworlds$, and we call $(\epsmodel,\world)$  a \emph{pointed epistemic model}.
Formally, given a pointed epistemic model $(\epsmodel,\world)$, we define the
semantics of $\langEL$ by induction on its formulas:
 $\epsmodel,\world\models p$ if $\world\in\valworlds(p)$,
 $\epsmodel,\world\models \neg\phi$ if it is not the case that $\epsmodel,\world\models\phi$,
 $\epsmodel,\world\models \phi\ou\psi$ if
    $\epsmodel,\world\models\phi$ or $\epsmodel,\world\models\psi$, and
 $\epsmodel,\world\models\Ki\phi$ if for all $\world'$ such
    that $\world\relworldsi\world'$, $\epsmodel,\world'\models\phi$.


\begin{definition}
\label{def-eventmodel}
An  \emph{event model} is a tuple
$\eventmodel=(\setevents,\{\releventsi\}_{i\in\agents},\pre,\post)$ where
 $\setevents$ is  finite set of \emph{events}, for each $i\in\agents$,
 $\releventsi\subseteq\setevents\times \setevents$
  is an \emph{accessibility
  relation} on $\setevents$ for agent $i$,
 $\pre: \setevents\to \langEL$ is a \emph{precondition function} and
 $\post: \setevents\to \AP \to \langEL$ is a \emph{postcondition function}.
\end{definition}

We write $\eventa\in \eventmodel$ for $\eventa\in \setevents$, and
call $(\eventmodel,\eventa)$ a \emph{pointed event model}. For an
event $e \in \eventmodel$, the precondition $\pre(e)$ and the
postconditions $\post(e)(p)$ ($p \in \AP$) are epistemic
formulas. They respectively describe  the set of worlds where event $e$ may take
place and the set of worlds where proposition $p$ will hold after
event $e$ has occurred. 

\begin{definition}
\label{def-prop-event}
A \emph{proposition event model} is an event model whose preconditions
and postconditions all lie
in the propositional fragment of $\langEL$.
\end{definition} 

We now define the \emph{update product} which, given an epistemic
model $\epsmodel$ and
an event model $\eventmodel$,  builds the epistemic model
$\epsmodel\otimes\eventmodel$ that represents the new epistemic
situation after $\eventmodel$ has occurred in $\epsmodel$. 

\begin{definition}
\label{def-update}
Let $\epsmodel=(\setworlds,\{\relworldsi\}_{i\in\agents},\valworlds)$
be an epistemic model and 
$\eventmodel=(\setevents,\{\releventsi\}_{i\in\agents},\pre,\post)$ be an event model.
The \emph{update product} of $\epsmodel$ and $\eventmodel$ is the
epistemic model $\epsmodel\otimes
\eventmodel=(\setworlds^\otimes,\{\relworldsi^\otimes\}_{i\in\agents},\valworlds^\otimes)$,
where $\setworlds^\otimes =\{(\world,\eventa)\in \setworlds\times
\setevents \mid \epsmodel,\world\models \pre(\eventa)\}$,
$\relworldsi^\otimes(\world,\eventa) =\{(\world',\eventa')\in
\setworlds^\otimes\mid \world'\in \relworlds_i(\world)\mbox{ and
}\eventa'\in \relevents_i(\eventa)\}$, and $\valworlds^\otimes(p)
=\{(\world,\eventa)\in \setworlds^\otimes\mid \epsmodel,\world\models
\post(\eventa)(p)\}$.
\end{definition}

The update product of a pointed epistemic
model $(\epsmodel,\world)$ with a pointed event model
$(\eventmodel,\eventa)$ is $(\epsmodel,\world)\otimes
(\eventmodel,\eventa)=(\epsmodel\otimes\eventmodel, (\world,\eventa))$
if $\epsmodel,\world\models\pre(\eventa)$, and it is undefined
otherwise.

To finish with this section, we define the \emph{size} of an epistemic model
$\epsmodel=(\setworlds,\{\relworldsi\}_{i\in\agents},\valworlds)$,
denoted by $|\epsmodel|$, as its number of edges:
$|\epsmodel|=\sum_{i\in\agents}|\relworldsi|$. The size of an event
model
$\eventmodel=(\setevents,\{\releventsi\}_{i\in\agents},\pre,\post)$,
that we note $|\eventmodel|$, is
its number of edges plus the sizes of precondition and postcondition
formulas:
$|\eventmodel|=\sum_{i\in\agents}|\releventsi|+\sum_{\eventa\in\setevents}(|\pre(\eventa)|
+ \sum_{p\in \AP}|\post(\eventa)(p)|)$.


\section{Trees, forests and $\CTLsKn$}
\label{sec-rlang}

 A \emph{tree alphabet} is a finite set of
\emph{directions}  $\Dirtree=\{d_1,d_2\ldots\}$. 
A \emph{$\Dirtree$-tree}, or
\emph{tree} for short when $\Dirtree$ is clear from the context, is a
set of words $\tree\subseteq \Dirtree^+$ that is closed for nonempty
prefixes, and for which
 there is a  direction $\racine=\tree\cap\Dirtree$, called the \emph{root}, such that  for all
    $\noeud\in\tree$, $\noeud=\racine\cdot\noeud'$ for some $\noeud'\in\Dirtree^*$.
A \emph{$\Dirtree$-forest}, or \emph{forest} when $\Dirtree$ is
understood, is defined likewise, except that it can have several
roots. Alternatively a forest can be seen as a union of trees.


    We classically allow nodes of trees and forests to carry
    additional information via labels: given a \emph{\labeling
      alphabet} $\Sigma$ and a tree alphabet $\Dirtree$, a
    \emph{$\Sigma$-\labeled $\Dirtree$-tree}, or
    \emph{$(\Sigma,\Dirtree)$-tree} for short, is a pair
    $\ltree=(\tree,\lab)$, where $\tree$ is a $\Dirtree$-tree and
    $\lab:\tree \rightarrow \Sigma$ is a \emph{\labeling}. 
The notion of
\emph{$(\Sigma,\Dirtree)$-forest} $\lforest=(\forest,\lab)$ is defined
likewise. Note that we use forests to represent the universe (to be defined) in
the semantics of $\CTLsKn$, hence the
notations $\lforest$ and $\forest$. 
Given a $\Dirtree$-forest $\forest$ and a node $\noeud=\dir_1\ldots\dir_n$ in the
forest $\forest$, we define the tree $\forest_\noeud$ to which this
node belongs as 
 the ``greatest'' tree in the forest $\forest$ that contains the
node $\noeud$: $\forest_\noeud=\{\noeuda\in\forest\mid
\dir_1\pref\noeuda\}$.  Similarly, given a
$(\Sigma,\Dirtree)$-forest $\lforest=(\forest,\lab)$ and a node
$\noeud\in\forest$, $\lforest_\noeud=(\forest_\noeud,\lab_\noeud)$, where
$\forest_\noeud$ is as above and $\lab_\noeud$ is the restriction of $\lab$ to the tree $\forest_\noeud$.


The set of well-formed $\CTLsKn$  formulas 
 is
 given by the following grammar:
\begin{align*}\mbox{State formulas:\hspace{1cm}}&\phi ::= p \mid \neg \phi \mid \phi\ou \phi \mid \A
  \psi \mid \Ki \phi & (\text{where } p\in\AP \text{ and } i\in\agents) \\ 
  \mbox{Path formulas:\hspace{1cm}}&\psi ::= \phi \mid \neg \psi \mid
  \psi \ou \psi \mid \X \psi \mid \psi \until \psi, \end{align*}

Let $\Dirtree$ be a finite set of directions, and let $\Sigma=2^{\AP}$
be the set of possible valuations. 
A $\CTLsKn$ (state) formula is interpreted in a node
of a $(\Sigma,\Dirtree)$-tree, but the semantics is
parameterized by, first, for each agent $i\in\agents$, a binary
relation $\reli$ between finite words over $\Sigma$, and second, a
forest of $(\Sigma,\Dirtree)$-trees which we see as the
\emph{universe}.  Preliminary to defining the semantics of  $\CTLsKn$, we let the \emph{node word}
of a node $\noeud=\dir_1\dir_2\ldots \dir_n \in \tree$ be
$\nword(\noeud) = \lab(\dir_1)\lab(\dir_1\dir_2)\ldots
\lab(\dir_1\ldots \dir_n) \in \Sigma^*$, made of the sequence of labels of all
nodes from the root to this node. Now, given a family
$\{\rel_i\}_{i\in\agents}$ of binary relations over $\Sigma^*$, a
$(\Sigma,\Dirtree)$-forest $\universe$, two nodes
$\noeud,\noeuda\in\universe$ and $i\in\agents$, we let $\noeud\reli
\noeuda$ denote that $\nword(\noeud)\reli\nword(\noeuda)$. 

A state formula of $\CTLsKn$ is interpreted over a
$(\Sigma,\Dirtree)$-tree $\ltree=(\tree,\lab)$ in a node
$\noeud\in\tree$, with an implicit universe $\universe$ and relations
$\{\reli\}_{i\in\agents}$, usually clear from the context: the notation $\ltree,\noeud\models\phi$
means that $\phi$ holds at the node $\noeud$ of the \labeled tree
$\ltree$.  Because all inductive cases but the knowledge operators
follow the classic semantics of $\CTLs$ on trees, 
we only give the semantics for  formulas of the form
$\Ki\phi$: 
\[\ltree,\noeud\models \Ki\phi   \hspace{.5cm}\mbox{ if for all } 
   \noeuda\in\universe \mbox{ such that } \noeud\rel \noeuda,\;
  \universe_\noeuda,\noeuda\models\phi \;\footnotemark[1]\]

\footnotetext[1]{Recall that
    $\universe_\noeuda$ is the biggest tree in $\universe$ that
    contains $\noeuda$
    .}

We shall use the notation $\ltree \models \phi$
 for $\ltree,\racine \models \phi$, where $\racine$ is the root of
 $\ltree$. 

Before stating the problems considered and our results, we establish in
the next section a connection between DEL-generated models and
regular structures, that allows us to apply automata techniques to
planning problems in DEL.

\section{DEL-generated models and regular structures}
\label{automatic}

We first briefly recall some basic definitions and facts concerning
finite state automata and transducers.  
A \emph{deterministic word automaton}  is a tuple
$\wA=(\Sigma,\wQ,\wdelta,\wq_\init,\wF)$, where $\Sigma$ is an \emph{alphabet}, $\wQ$ is a finite set of
\emph{states}, $\wdelta:\wQ\times\Sigma\to\wQ$ is a partial  \emph{transition function}
and $\wF$ is a set of \emph{accepting} states. 
The \emph{language} accepted by  a word automaton $\wA$
consists in the set of words accepted by $\wA$, and it is classically written $\lang(\wA)$. It is well known that
the set of languages accepted by word automata is
precisely the set of \emph{regular} word
languages. 
A \emph{finite state synchronous transducer}, or
\emph{synchronous transducer} for short, is a finite word
automaton with two tapes, that reads one letter from each tape
at each transition. Formally, a synchronous transducer is a tuple
$\trans=(\Sigma,\transQ,\transDelta,\transq_\init,\transF)$,
where the components are as for word automata, except for the
\emph{transition relation}
$\transDelta\subseteq\transQ\times\Sigma\times\Sigma\times\transQ$. The
(binary) relation recognized by a transducer $\trans$ is denoted by
$[\trans]\subseteq \Sigma^*\times\Sigma^*$. Synchronous transducers are known to recognize
the set of \emph{regular relations}, also called \emph{synchronized
  rational relations} in the literature (see \cite{elgot1965relations,berstel1979transductions,DBLP:journals/corr/abs-1304-4150}).
In the following, the size of a transducer $\trans$, written $|\trans|$, will denote the
size of its transition relation: $|\trans|=|\transDelta|$.

\begin{definition}
A \emph{relational structure} is a tuple $\str=(\strdom,\{\strreli\}_{i\in \agents},\strval)$ where $\strdom$
is the (possibly infinite) \emph{domain} of $\str$, for each $i\in\agents$,
$\reli\;\subseteq \strdom\times\strdom$ is a binary relation and
$\strval:\AP\to 2^\strdom$ is a valuation function. 
$\strval$ can alternatively be seen as a set
of predicate interpretations for atomic propositions in $\AP$. 
\end{definition}

\begin{definition}
\label{def-str}
A relational structure $\str=(\strdom,\{\strreli\}_{i\in
  \agents},\strval)$ is a \emph{regular structure} over a finite
alphabet $\Sigma$ if its domain $\strdom\subseteq \Sigma^*$ is a
regular language over $\Sigma$, for each $i$, $\reli\;\subseteq
\Sigma^*\times\Sigma^*$ is a regular relation and for each $p\in\AP$,
$\strval(p)\subseteq \strdom$ is a regular language.  Given
deterministic word automata $\wauto_\str$ and $\wauto_p$
($p\in\AP$), as well as transducers $\trans_i$ for $i\in\agents$, we
say that $(\wauto_\str,\{\trans_i\}_{i\in\agents},\{\wauto_p\}_{p\in\AP})$
is a \emph{representation} of $\str$ if $\lang(\wauto_\str)=\strdom$, for each
$i\in\agents$, $[\trans_i]=\;\reli$ and for each $p\in\AP$,
$\lang(\wauto_p)=\strval(p)$.
\end{definition} 

\begin{definition}
\label{def-DELforest}
For an epistemic model
$\epsmodel=(\setworlds,\{\relworldsi\}_{i\in\agents},\valworlds)$ and
an event model
$\eventmodel=(\setevents,\{\releventsi\}_{i\in\agents},\pre,\post)$,
we define the family of epistemic models $\{\ETLforest[n]{\epsmodel}{\eventmodel}\}_{n\geq 0}$
by letting $\ETLforest[0]{\epsmodel}{\eventmodel}=\epsmodel$ and
$\ETLforest[n+1]{\epsmodel}{\eventmodel}=\ETLforest[n]{\epsmodel}{\eventmodel}\otimes\eventmodel$. Letting,
for each $n$,
$\ETLforest[n]{\epsmodel}{\eventmodel}=(\setworlds^n,\{\relworldsi^n\}_{i\in\agents},\valworlds^n)$,
we define the relational structure generated by $\epsmodel$ and
$\eventmodel$ as $\ETLforest{\epsmodel}{\eventmodel}=
(\strdom,\{\strrel_i\}_{i\in\agents},\strval)$, where:
\begin{itemize}
  \item $\strdom=\bigunion_{n\geq 0}\setworlds^n$,
  \item $\hist\strrel_i\hist'$ if there is some $n$ such that
    $\hist,\hist'\in\ETLforest[n]{\epsmodel}{\eventmodel}$ and $\hist \relworldsi^n \hist'$, and
  \item $\strval(p)=\bigunion_{n\geq 0}\valworlds^n(p)$.
\end{itemize}
\end{definition}

\begin{proposition}
\label{prop-merging}
If $\epsmodel$ is an epistemic model and $\eventmodel$ is a propositional event
model, then $\ETLforest{\epsmodel}{\eventmodel}$ is a regular
structure, and it admits a representation of  size
$2^{O(|\AP|)}\cdot (|\epsmodel| + |\eventmodel|)^{O(1)}$.
\end{proposition}

\begin{proof}
  Let $\epsmodel=(\setworlds,\relworlds,\valworlds)$ be an epistemic
  model, let 
  $\eventmodel=(\setevents,\relevents,\pre,\post)$ be a propositional event model, and 
  let
  $\ETLforest{\epsmodel}{\eventmodel}=(\strdom,\{\strrel_i\}_{i\in\agents},\strval_\strdom)$.

  Define the word
  automaton $\wauto_\strdom=(\Sigma,\wQ,\wdelta,\wq_\init,\wF)$, where $\Sigma=\setworlds\union\setevents$,
   $\wF=\{\wq_\val\mid \val\subseteq \AP\}$ and $\wQ=\wF\uplus\{\wq_\init\}$. 
  For a world $\world\in\setworlds$,
  we define its \emph{valuation} as $\val(\world)\egdef\{p\in\AP\mid
  \world\in\valworlds(p)\}$.
  We now define $\wdelta$, which
  is the following partial transition function:
  \[
  \begin{array}{ll}
    \forall \world\in\setworlds$, $\forall\eventa\in\setevents, &\\
    \wdelta(\wq_\init,\world)=\wq_{\val(\world)} &
    \wdelta(\wq_\init,\eventa) \mbox{ is undefined,}\\
    \wdelta(\wq_\val,\world)  \mbox{ is undefined} &
    \wdelta(\wq_\val,\eventa) = 
    \begin{cases} 
      \wq_{\val'}\mbox{, with }\val'=\{p\mid \val\models\post(\eventa)(p)\}
      & \mbox{if }\val\models\pre(\eventa)\\
      \mbox{undefined} & \mbox{otherwise.}
    \end{cases}
  \end{array}
  \]
It is not hard to see that $\lang(\wauto_\strdom)=\strdom$, hence
$\strdom$ is a regular language. Also, $\wauto_\strdom$ has
$2^{|\AP|}+1$ states, and each state has at most
$|\epsmodel|+|\eventmodel|$ outgoing transitions, so that
$|\wauto_\strdom|=2^{O(|\AP|)}\cdot (|\epsmodel|+|\eventmodel|)$.


Concerning  valuations, take some $p\in\AP$. 
Let $\wauto_p=(\Sigma,\wQ,\wdelta,\wq_\init,\wFp)$, where $\wFp=\{\wq_\val\mid
p\in\val\}$. Clearly, $\lang(\wauto_p)=\strval_\strdom(p)$, hence $\strval_\strdom(p)$
is a regular language, and $|\wauto_p|=|\wauto_\strdom|$. 


For the relations, let $i\in\agents$ and consider the one-state
synchronous transducer
$\trans_i=(\Sigma,\transQ',\transDelta_i,\wq_\init,\transF')$,
where
$\transQ'=\{\transq\}$, $\wq_\init=\wq$, $\transF'=\{\transq\}$, and
$\transDelta_i=\{(\transq,\world,\world',\transq)\mid\world\relworlds_i\world'\}\cup\{(\transq,\eventa,\eventa',\transq)\mid\eventa\relevents_i\eventa'\}$.
It is easy to see that $\strrel_i\;=[\trans_i]\inter
\strdom\times\strdom$. Since $[\trans_i]$ is a regular relation and
$\strdom$ is a regular language, $\strrel_i$ is a regular relation
recognized by
$\transi'=\trans_{\strdom}\compo\transi\compo\trans_{\strdom}$,
where $\trans_{\strdom}$ is a synchronous transducer that
recognizes the identity relation over $\strdom$ (easily obtained from
$\wauto_\strdom$). 
This transducer is of size $|\transi'|=|\trans_{\strdom}|^2\cdot
|\transi|= 2^{O(|\AP|)}\cdot (|\epsmodel|+|\eventmodel|)^{O(1)}$.
Finally, $\ETLforest{\epsmodel}{\eventmodel}$ is a regular
structure that accepts
$(\wauto_\strdom,\{\transi'\}_{i\in\agents},\{\wauto_p\}_{p\in\AP})$
as a regular representation of size $2^{O(|\AP|)}\cdot (|\epsmodel| +
|\eventmodel|)^{O(1)}$. One can check that this is also an upper bound
on the time
needed to compute this representation.

\end{proof}


\section{Epistemic protocol synthesis}
\label{protocol}

We first consider the   problem of epistemic
planning
\cite{DBLP:journals/jancl/BolanderA11,DBLP:conf/lori/LowePW11} studied
in the Dynamic Epistemic Logic community. Note that our formulation
slightly differs from the classic one as we consider a unique event model,
but both problems can easily be proved inter-reducible in linear time.

\renewcommand{\subsetevents}{\setevents}

\begin{definition}[Epistemic planning problem]
\label{def-onemodelplanning}
  Given a pointed
  epistemic model $(\epsmodel_\init,\world_\init)$, an event model $\eventmodel$,
  a set of events $\subsetevents\subseteq \eventmodel$
  and a goal formula $\phi\in\langEL$,  decide if there
  exists a finite series of events
  $\eventa_1\ldots\eventa_n$ in $\subsetevents$
  such that
  $(\epsmodel_\init,\world_\init)\otimes(\eventmodel,\eventa_1)\otimes\ldots\otimes(\eventmodel,\eventa_n)\models\phi$.
  The \emph{propositional epistemic planning problem} is the
  restriction of the  epistemic planning problem to
  propositional event models.
\end{definition}


The epistemic planning problem is undecidable 
\cite{DBLP:journals/jancl/BolanderA11,DBLP:conf/ijcai/AucherB13}. However,
\cite{DBLP:journals/jancl/BolanderA11} proved that the problem is
decidable in the case of one agent and equivalence accessibility
relations in epistemic and event models. More recently,
 \cite{DBLP:conf/ijcai/AucherB13} and
 \cite{Yu2013}  proved independently that the one agent problem is
 also decidable for $\treu$ accessibility relations. \cite{Yu2013} also
proved that restricting to propositional event models yields
decidability of the epistemic planning problem, even for several
agents and arbitrary accessibility relations.


\begin{theorem}[\cite{Yu2013}]
\label{th-chinois}
The propositional epistemic planning problem is decidable. 
\end{theorem}

 Proposition~\ref{prop-merging} allows us to establish an alternative
proof of this result, with two side-benefits. First, using automata
techniques,  our decision procedure can
synthesize as a by-product a finite word automaton that generates exactly the
(possibly infinite) set of all
solution plans. Second, we obtain accurate upper-bounds on the time complexity.

For an instance $(\epsmodel,\eventmodel,\setevents,\phi)$ of the
epistemic planning problem, we define its size as the sum of its
components' sizes, plus the number of atomic propositions:
$|\epsmodel,\eventmodel,\setevents,\phi|=|\epsmodel|+|\eventmodel|+|\setevents|+|\phi|+|\AP|$.

\begin{theorem}
\label{theo-synth-epistemic-planning}
The propositional epistemic planning problem is in $\kEXPTIME[k+1]$ for
formulas of nesting depth $k$. 
Moreover, it is possible to build in the same time a finite word
automaton $\wautoplans$ such that $\lang(\wautoplans)$ is the set of
all solution plans.
\end{theorem}

\begin{proofsketch}
 Let $(\epsmodel,\eventmodel,\setevents,\phi)$ be an instance of the
 problem.  By Proposition~\ref{prop-merging} we obtain an exponential
 size automatic
representation of the forest $\ETLforest{\epsmodel}{\eventmodel}$: the
set of possible histories, as well as their valuations, are represented by a finite automaton $\wauto$, and
the epistemic relations are given by finite state transducers. 
Because the
epistemic relations are rational, we can use the powerset
construction presented in \cite{maubertIGTR2013} in the context of
uniform strategies 
\cite{maubertIGTR2013,maubertSR2013,maubertFSTTCS2013court}. 
Indeed, this construction easily generalizes  to the case of $n$ relations, and even
though in \cite{maubertIGTR2013} it is defined on game arenas it can,
in our context,  be adapted
to regular structures.
Letting $k$ be the maximal nesting depth of knowledge
operators in $\phi$, this construction yields an automaton $\pow{\wauto}$ of size
$k$-exponential in the size of $\wauto$, hence $(k+1)$-exponential in
$|\epsmodel,\eventmodel,\setevents,\phi|$, that still represents
$\ETLforest{\epsmodel}{\eventmodel}$, and in which $\phi$ can be evaluated
positionally. Keeping only  transitions
 \labeled by events in $\subsetevents$, and choosing for accepting states
those that verify $\phi$, we obtain the automaton $\wautoplans$ that
recognizes the set of solution plans. Furthermore, solving the epistemic
planning problem amounts to solving the nonemptiness problem for
$\lang(\wautoplans)$; this can be done in time linear in the size of
$\wautoplans$, which is $k+1$-exponential in the size of the input $(\epsmodel,\eventmodel,\subsetevents,\phi)$.  
\end{proofsketch}

In fact, the correspondence between the DEL framework and automatic
structures established in Proposition~\ref{prop-merging} allows us to
solve a much more general problem than  epistemic planning. 

We generalize the notion of epistemic planning in three
directions. First, we no longer consider finite sequences of actions
but infinite ones. As a consequence, we need not stick to reachability
objectives as in planning (where the aim is to reach a state of the
world that verifies some formula), and we therefore allow for any epistemic
temporal formula as objective, which is the second
generalization. Finally, we no longer look for a single series of
events, but we try to synthesize a \emph{protocol}, \ie a set of
plans.

\begin{definition}
\label{def-protocol}
Given an epistemic model $\epsmodel$ and an event model $\eventmodel$, an
\emph{epistemic protocol} is a forest
$\proto\subseteq\ETLforest{\epsmodel}{\eventmodel}$; it is  \emph{rooted} if it is a
tree.
\end{definition}

\begin{definition}[Epistemic protocol synthesis problem]
\label{def-proto-synth}
  Given an initial pointed epistemic model $(\epsmodel,\world)$, a propositional event model $\eventmodel$
  and a $\CTLsKn$ formula $\phi$, letting
  $\universe=\ETLforest{\epsmodel}{\eventmodel}$ be the universe, decide if there is an
  epistemic protocol $\proto\subseteq\universe$
  rooted in $\world$ such that $\proto\models\phi$, and synthesize such a
  protocol if any.
\end{definition}

Again making use of Proposition~\ref{prop-merging}, the epistemic
protocol synthesis  problem can be reduced to synthesizing a
uniform strategy in a game arena with regular relations between
plays. This  can be solved with the
powerset construction from \cite{maubertIGTR2013} and classic automata techniques for solving
games with $\CTLs$ winning condition. We finally obtain the following
result. 

\begin{theorem}
\label{theo-proto-synth}
The epistemic protocol synthesis  problem is decidable.
If  the
 nesting depth of the  goal formulas is bounded by $k$, then the
problem is in \kEXPTIME[\max(2,k+1)]. 
\end{theorem}

\section{Discussion}

We have described a connection between DEL-generated models and
regular structures, which enabled us to resort to a combination of
mature automata techniques and more recent ones developed for the
study of uniform strategies, in order to solve planning problems in
the framework of DEL. We believe that this is but a first step in
applying classic automata techniques developed for temporal logics to
the study of dynamic epistemic logic.  As witnessed by classic works on automata-based program
synthesis (see for example \cite{pnueli89b,DBLP:conf/stacs/Thomas95}),
 automata techniques are well suited to tackle problems such as
synthesizing plans, protocols, strategies or programs,
and we  believe
that it should also be the case in the DEL framework;  in
addition the complexity of solving classic automata problems such as
nonemptiness has been extensively studied, and this may help to settle
the complexity of problems in DEL, such as the epistemic
planning problem.

As for future work, we would like to investigate the optimality of the
upper-bounds that we obtained on the time complexity of the
epistemic planning problem for propositional event models, as well as
for our notion of epistemic protocol synthesis.  
Another direction for future research concerns the latter problem: a next step would be to apply
techniques from control theory and quantified $\mu$-calculus \cite{DBLP:conf/mfcs/RiedwegP03} to synthesize \emph{maximal permissive}
epistemic protocols. In general such objects only exist
for safety objectives, but recently a weaker notion of
\emph{permissive strategy} has
been studied in the context of parity games
\cite{DBLP:journals/ita/BernetJW02}. A  strategy is permissive if it
contains the behaviours of all memoryless strategies, and such
strategies always exist in parity games.  Similar notions may be
introduced for protocols with epistemic temporal objectives to capture 
concepts of ``sufficiently permissive'' protocols.



\bibliographystyle{eptcs}
\bibliography{../../Biblio/games,../../Biblio/logic,../../Biblio/opacity,../../Biblio/misc,../../Biblio/classic,../../Biblio/automata,../../Biblio/transducers}



\end{document}